\documentclass[aps,pra,twocolumn,superscriptaddress,floatfix,nofootinbib,showpacs,longbibliography]{revtex4-2}

\usepackage[utf8]{inputenc}  
\usepackage[T1]{fontenc}     
\usepackage[table]{xcolor}
\usepackage[british]{babel}  
\usepackage[sc,osf]{mathpazo}
\usepackage[scaled=0.86]{berasans}  
\usepackage[colorlinks=true, citecolor=blue, urlcolor=blue]{hyperref}  
\usepackage{graphicx} 
\usepackage[babel]{microtype}  
\usepackage{amsmath,amssymb,amsthm,bm,amsfonts,mathrsfs,bbm} 

\usepackage{xspace}  
\usepackage{pgf,tikz}
\usepackage{xcolor}
\usepackage{multirow}
\usepackage{array}
\usepackage{bigstrut}
\usepackage{braket}
\usepackage{color}
\usepackage{natbib}
\usepackage{multirow}
\usepackage{mathtools}
\usepackage{float}
\usepackage[caption = false]{subfig}
\usepackage{xcolor,colortbl}
\usepackage{color}
\usepackage{rotating}
\usepackage{tikz}
\usepackage{ulem}
\usetikzlibrary{quantikz}

\newcommand{\be}{\begin{equation}}
\newcommand{\ee}{\end{equation}}
\newcommand{\ba}{\begin{eqnarray}}
\newcommand{\ea}{\end{eqnarray}}

\newtheorem{theorem}{Theorem}

\newtheorem{definition}{Definition}
\newtheorem{proposition}{Proposition}

\newtheorem{lemma}{Lemma}

\def\>{\rangle}
\def\<{\langle}

\usepackage{centernot}
\usepackage{subfig}

\begin{document}

\title{ Conclusive Local State Marking: More Nonlocality With No Entanglement}

\author{Samrat Sen}
\affiliation{Department of Physics of Complex Systems, S. N. Bose National Center for Basic Sciences, Block JD, Sector III, Salt Lake, Kolkata 700106, India.}

\begin{abstract}

Nonlocality exhibited by ensembles of composite quantum states, wherein local operations and classical communication (LOCC) yield suboptimal discrimination probabilities compared to global strategies, is one of the striking nonclassical features of quantum theory. A variant of this phenomenon arises in conclusive local state discrimination, where the correct state must be identified with zero error, albeit allowing for inconclusive outcomes. More recently, the notion of local state marking has been introduced, with the focus shifted to correctly identifying the permutation of a subset of states randomly chosen from a given set of multipartite states under LOCC. In this work, we unify these two approaches by introducing the task of conclusive local state marking, which reveals a finer hierarchy of nonlocality in multipartite quantum state ensembles. Notably, we demonstrate that certain ensembles of product states can exhibit a stronger form of nonlocality than entangled ensembles traditionally considered highly nonlocal.
\end{abstract}


\maketitle	
{\it Introduction.} Quantum nonlocality—wherein entangled states shared between distant parties yield input-output correlations that defy any {\it local-causal} explanation—is one of the most striking features of quantum theory, as established by Bell’s theorem \cite{Bell1964,BellAPS,Mermin,BellReview}. In contrast, correlations arising from separable states are always factorizable and thus admit a local description. Interestingly, even ensembles of multipartite product states can exhibit a subtler form of nonlocality wherein the success probability of discriminating a state given from such an ensemble under local operations and classical communication (LOCC) can remain suboptimal as compared to the global discrimination probability. This phenomenon, termed `nonlocality without entanglement', was first highlighted by Peres and Wootters \cite{PeresWooters}, and gained prominence following the seminal work of Bennett {\it et al.} \cite{BennettNLWE}, who constructed orthogonal product-state bases that are locally indistinguishable. These findings have since spurred extensive research under the broader theme of local state discrimination (LSD),  yielding a rich array of results \cite{Bennett99(1),DiVincenzo03,Niset06,DuanProduct,Calsamiglia10,Bandyopadhyay,Chitambar14,Halder18,Demianowicz18,Halder19,Halder19(1),Agrawal19,Rout19,Bhattacharya20,Banik20,Rout20,LSM,Subhendu,Tathagata1,Pratik,Subhendu2}, owing to its relevance in fundamental implications for quantum information theory and applications in cryptographic protocols \cite{Gagliardoni2021Quantum,Bennett1984Quantum,Markham2008Graph,Matthews2009Distinguishability}. Over time, this has inspired several variants of the task. For instance, conclusive state discrimination (CSD) relaxes the requirement of always providing a definite answer, allowing for inconclusive outcomes while still ensuring correctness when a conclusion is drawn \cite{CheflesGlobal}, thereby making it less demanding than the task of perfect state discrimination. Another variant, known as local state marking (LSM), requires the parties to determine the specific permutation in which the quantum states are distributed among them \cite{LSM}. \setlength{\parindent}{15pt}   
\setlength{\parskip}{0pt}      

As shown by Duan and collaborators \cite{DuanProduct}, the celebrated product 
states introduced by Bennett \textit{et al.} \cite{BennettNLWE} lose their 
nonlocal property once conclusive local state discrimination (CLSD) is considered: these 
states become both globally and locally conclusively distinguishable. Duan 
\textit{et al.} further provided the first examples of product states that 
are nonlocal even in the CLSD paradigm, thereby establishing conclusive local 
indistinguishability as a strictly stronger manifestation of nonlocality than 
standard local indistinguishability.

 With this in mind, in this work, we revisit the task of LSM \cite{LSM}, which in its original version assumes mutually orthogonal states. We introduce a relaxed version called conclusive local state marking (CLSM) to allow for conclusively figuring out the permutation in which the set of states has been shared. For any set \( \mathcal{S} \) of multipartite quantum states, we define a family of state discrimination tasks referred to as \( m \)-CLSM, parameterized by an integer \( m \in \{1,..., |\mathcal{S}|\} \), where \( |S| \) denotes the cardinality of the set. The case \( m = 1 \) corresponds to the standard conclusive local state discrimination (CLSD) task \cite{Chefles}, while \( m = |\mathcal{S}| \) corresponds to the full conclusive local state marking task, which we refer to simply as CLSM.  This broader framework also accommodates linearly independent ensembles, thereby generalizing beyond ensembles of states containing only mutually orthogonal states. Thus, CLSM represents a task that is even less demanding than LSM, thereby expanding the landscape of LOCC-based discrimination paradigm. The CLSM task is shown to be an inequivalent task in comparison to CLSD. This we prove by showing that certain binary ensembles of mutually orthogonal quantum states—previously known to exhibit strong local indistinguishability \cite{Bandyopadhyay,Duan}—can, in fact, be conclusively marked using local operations and classical communication (LOCC). Even with access to an arbitrarily large number of copies, these ensembles were previously shown to be not only locally indistinguishable, but also conclusively indistinguishable under PPT (positive partial transpose) operations—a broader class that strictly contains LOCC operations. Notably, a set of states which is unmarkable, be it perfect or conclusive, shows a stronger form of nonlocality than what is considered in literature. In this direction, we give two examples of a set of four product states that are not only conclusively indistinguishable via LOCC, but whose \(2\)-CLSM task is also not possible. Thus, we present the first known example of product states that are conclusively unmarkable, thereby exhibiting a stronger form of nonlocality than the original example by Bennett {\it et al.}~\cite{BennettNLWE} or Peres and Wootters \cite{PeresWooters}.  We also give an example of a set of mutually orthogonal pure states, whose $2$-LSM as well as $2$-CLSM is also not possible. \\

While any set of mutually orthogonal quantum states can always be perfectly distinguished, the scenario becomes more intriguing in multipartite systems under restricted operational settings. A particularly relevant operational framework is the LOCC, where spatially separated parties perform local quantum operations on their respective subsystems, supplemented by rounds of classical communication. Due to the multi-round nature of such communication, characterizing the full set of LOCC operations remains a challenging task \cite{Chitambar2014Everything}. Remarkably, even orthogonal states may become indistinguishable under LOCC, giving rise to the phenomenon of `quantum nonlocality without entanglement' \cite{BennettNLWE}. Numerous examples of such constructions have since been reported, accompanied by detailed characterizations aimed at deepening our understanding of the area of local state discrimination \cite{Xu2016,Yang2015,Gao,Fei,Xu2017,Wang2017,Zhang,Zhang2015,Zheng,Zhang2016,Zhang2016Yongjun,Wang2017arxiv,Zhang2017Luo,Zhang2017Tan,Zuo2021,Zhang2024,Zhu2022,Zhang2021,Cao2025,Feng,Cohen,Croke,Childs2013,MassarPop,Tian,Rinaldis,Cosentino2013,Ghosh2004,Bandyopadhyay2013,Walgate2002,Fan,Nathanson2005,Watrous,Hayashi2006,Bandyopadhyay2011,Yu}.

However, it is known that no ensemble of just two mutually orthogonal pure states can exhibit such nonlocality \cite{walgatehardy}. However, when mixed states are considered in a binary ensemble, local distinguishability is no longer guaranteed. Indeed, there exist examples of two mutually orthogonal quantum states that are locally indistinguishable, even when multiple copies are provided \cite{Bandyopadhyay, Duan}. The key insight underlying these results is that the states cannot be distinguished even conclusively via LOCC, thereby ruling out perfect local discrimination as well. This underscores the power of conclusive local indistinguishability as a diagnostic tool for detecting local indistinguishability. In the following, we explore the concept of conclusive discrimination in greater depth.


Given a state $\ket{\psi_k}$, randomly chosen from an ensemble $\mathcal{S}\equiv\left\{ \ket{\psi_1}, \ket{\psi_2}, \dots, \ket{\psi_N} \right\}$, conclusive state discrimination (CSD) task demands unambiguous identification of the state with zero probability of error, albeit there is a nonzero chance that the protocol yields an inconclusive outcome. We thus require a measurement with \( N+1 \) outcomes: \( N \) of these correspond to the $N$ states that are to be conclusively discriminated, while the remaining outcome denotes an inconclusive result. This can be expressed in terms of positive operator-valued measures (POVMs), $\mathcal{M}\equiv\{P_?, P_k|~P_? + \sum_{k=1}^N P_k = \mathbb{1}_\mathcal{H}\}$. The POVM elements satisfy
\begin{equation}
\operatorname{Tr}(\ket{\psi_l}\bra{\psi_l} P_{k}) = p_k \delta_{l,k}~,\label{eq:POVMcondition}
\end{equation}
where \( p_k = \operatorname{Tr}(\ket{\psi_k}\bra{\psi_k} P_k) > 0 \) is the probability of correctly detecting the state \(\ket{\psi_k} \), and $P_?$ corresponds to inconclusive outcome. Thus, for example, even though the two nonorthogonal quantum states  $\ket{0}$ and $\ket{+}$, cannot be perfectly distinguished, it can be conclusively distinguished. 
Subsequently, a necessary and sufficient condition for conclusive distinguishability is obtained. \begin{lemma}\label{lemmaCheflesGlobal}
[Chefles \cite{CheflesGlobal}] A set of 
$N$ quantum states can be conclusively distinguished if and only if the states are linearly independent. 
\end{lemma}
For instance, the set \( \{\ket{0}, \ket{+}, \ket{1}\} \subset \mathbb{C}^2 \) cannot be conclusively distinguished. Notably, within the LOCC paradigm the concept of conclusive distinguishability get further refined due to constrain on the allowed operations.
\begin{definition}\label{defCheflesCLSD}
[Chefles \cite{Chefles}] A multipartite state \( |\psi\rangle \in S \) is conclusively locally identifiable if and only if there is an 
LOCC protocol whereby with some nonzero probability \( p > 0 \) it can be determined that the multipartite quantum system \( Q \) 
was certainly prepared in state \( |\psi\rangle \). The set \( S \) is conclusively locally distinguishable if and only if every state  in \( S \) is conclusively locally identifiable.
\end{definition}
Notably, in the multipartite case local conclusive distinguishability has a crucial difference from the global conclusive distinguishability. In the local case one need to ensure LOCC protocol for each and every member of the set, which, unlike the global case, may vary from state to state. Mere linear independence does not ensure local conclusive distinguishability of a set of multipartite states. For instance, consider the four Bell states in $\mathbb{C}^2 \otimes\mathbb{C}^2$:
\begin{equation*}
\mathcal{S}^{\text{Bell}} \equiv 
\left\{
\begin{aligned}
\ket{\mathcal{B}^1} &:= \ket{\Phi^+}_{AB} := (\ket{00} + \ket{11})/\sqrt{2} \\
\ket{\mathcal{B}^2} &:= \ket{\Phi^-}_{AB} :=(\ket{00} - \ket{11})/\sqrt{2} \\
\ket{\mathcal{B}^3} &:= \ket{\Psi^+}_{AB} :=(\ket{01} + \ket{10})/\sqrt{2} \\
\ket{\mathcal{B}^4} &:= \ket{\Psi^-}_{AB} :=(\ket{01} - \ket{10})/\sqrt{2}
\end{aligned}
\right\}.
\end{equation*}
These states are mutually orthogonal and therefore linearly independent. However, these states are known to be locally indistinguishable \cite{GKar}. The proof relies on the separability property of the Smolin state \cite{Smolin}, given by  
\begin{align}
\rho_S = \frac{1}{4} \sum_{i=1}^{4} \mathcal{B}^i_{A_1B_1} \otimes \mathcal{B}^i_{A_2B_2} = \frac{1}{4} \sum_{i=1}^{4} \mathcal{B}^i_{A_1A_2} \otimes \mathcal{B}^i_{B_1B_2},    
\end{align}
along with the fundamental fact that entanglement cannot be increased under LOCC [See \cite{Chitambar2014Everything} for an in-depth understanding]; \(\mathcal{B}\equiv\ket{\mathcal{B}}\bra{\mathcal{B}}\). The second decomposition shows that the state is separable across the bipartition \( A_1A_2 : B_1B_2 \), implying that Alice and Bob, who hold the \( A_1A_2 \) and \( B_1B_2 \) subsystems respectively, share no entanglement. In contrast, the first decomposition reveals that if the Bell states are locally distinguishable, then their local distinguishability would result in entanglement between Alice and Bob -- a contradiction. Since it is impossible to generate entanglement from a separable state via LOCC, even probabilistically, a similar argument implies that Bell states cannot be conclusively distinguished using LOCC.

On the other hand, consider the nine orthogonal product states of Bennett \textit{et al.} \cite{BennettNLWE}:
\renewcommand{\arraystretch}{1.3}
\begin{equation*}
\mathcal{S}^{\text{B}} \equiv \left\{
\begin{array}{ll}
\ket{\psi}_1 = \ket{1}_A\ket{1}_B,~~~~~\ket{\psi}_2 = \ket{0}_A\ket{0{+}1}_B,\\
\ket{\psi}_3 = \ket{0}_A\ket{0{-}1}_B,~\ket{\psi}_4 = \ket{2}_A\ket{1{+}2}_B,\\
\ket{\psi}_5 = \ket{2}_A\ket{1{-}2}_B,~\ket{\psi}_6 = \ket{1{+}2}_A\ket{0}_B,\\
\ket{\psi}_7 = \ket{1{-}2}_A\ket{0}_B,~\ket{\psi}_8 = \ket{0{+}1}_A\ket{2}_B,\\
\hspace{1.5cm}\ket{\psi}_9 = \ket{0{-}1}_A\ket{2}_B, \\
\end{array}
\right\},
\end{equation*}
\renewcommand{\arraystretch}{1.0}
where $\ket{p \pm q} := (\ket{p} \pm \ket{q})$. Although this set is locally indistinguishable, they turn out to be conclusively distinguishable under LOCC. To see this, take the state \(\ket{\psi_1}\in\mathcal{S}^B\). Alice and Bob can  conclusively identify  \(\ket{\psi_1}\in\mathcal{S}^B\) through a local protocol that consists of computational basis measurement. Whenever the product effect \(\ket{1}\bra{1}\otimes\ket{1}\bra{1}\) clicks they are sure the given state was \(\ket{\psi_1}\), else they register an inconclusive outcome. Similar protocols can be constructed for all the other members. Subsequently, this observation has been generalized to obtain a necessary and sufficient criterion for conclusive local discrimination. 
\begin{lemma}\label{lemmaCheflesProductDetecting}
[Chefles \cite{Chefles}]Let \( S = \{ |\psi_1\rangle, \dots, |\psi_m\rangle \} \) be a collection of \( m \) 
quantum states on \( \mathcal{H} \). Then \( S \) can be conclusively discriminated by LOCC if and only if for each \( 1 \leq k \leq m \), there exists a product detecting state \( |\phi_k\rangle \) such that \( \langle \psi_j | \phi_k \rangle = 0 \) for \( j \neq k \) and \( \langle \phi_k | \phi_k \rangle \neq 0 \).   
\end{lemma}
Notably, nonlocality without entanglement as identified in \cite{BennettNLWE} has an intriguing connection with unextendible product bases (UPB) \cite{Bennett99(1)}, which subsequently has lead to the notion of unextendible bases (UB) \cite{DuanProduct}.
\begin{definition} 
A set $\mathcal{S} = \{\ket{\phi_1}, \ldots, \ket{\phi_m}\}$ of $m$ linearly independent quantum states on a Hilbert space $\mathcal{H}$ is said to be an unextendible if its orthogonal complement $\mathcal{S}^\perp$ contains no product state; otherwise, $\mathcal{S}$ is said to be extendible. Furthermore, $\mathcal{S}$ is said to be a genuinely unextendible basis (GUB) if it is unextendible and any proper subset of $\mathcal{S}$ is extendible. 
\end{definition}\label{defUB}
An UB will be called UPB if it contains only product states, and it will be called GUPB if it further satisfies the requirement of genuineness. Notably, a necessary and sufficient criteria for conclusive local state discrimination of an UB is obtained subsequently.
\begin{lemma}\label{lemmaGUB}
[Duan et al. \cite{DuanProduct}]
An UB allows conclusive local discrimination if and only if it is a GUB.
\end{lemma} 
The authors in \cite{DuanProduct} have also provided an example of a set of pure product states in \(\mathbb{C}^2\otimes\mathbb{C}^2\) which are not conclusively distinguishable via LOCC :
\begin{align}
\mathcal{S}^{\text{D}} \equiv \left\{
\begin{aligned}
  &\ket{D_1} := \ket{0}_A \ket{0}_B, \quad 
   \ket{D_2} := \ket{1}_A \ket{1}_B, \\
  &\ket{D_3} := \ket{+}_A \ket{+}_B,\, \ket{D_4} := \ket{i_+}_A \ket{i_-}_B
\end{aligned}
\right\}, \label{Duanproductstates}
\end{align}

where $\ket{\pm}:=1/\sqrt{2}(\ket{0} \pm\ket{1})$, $\ket{i_{\pm}}:=1/\sqrt{2}(\ket{0} \pm\iota\ket{1})$, with $\iota:=\sqrt{-1}$

 In terms of conclusive local distinguishability, the set $\mathcal{S}^{\text{D}}$ exhibits a stronger form of nonlocality than $\mathcal{S}^{\text{B}}$. However, this comparison may appear less meaningful, as $\mathcal{S}^{\text{D}}$ comprises non-orthogonal states, whereas the states in $\mathcal{S}^{\text{B}}$ are mutually orthogonal. A more relevant comparison can be made with the well-known double-trine ensemble in $\mathbb{C}^2 \otimes \mathbb{C}^2$, introduced by Peres and Wootters \cite{PeresWooters} (see also \cite{ChitambarP-W}):
\begin{align}
\mathcal{S}^{\text{P-W}} &\equiv \left\{ \ket{w_{k}}_A \ket{w_{k}}_B \right\}_{k=0}^{2}, \nonumber\\
\text{where}~\ket{w_k} &:= \exp\left(-\frac{k\pi}{3} \sigma_y\right) \ket{0}. \label{eqPW}
\end{align}

This ensemble is highly symmetric and exemplifies `nonlocality without entanglement', as its optimal discrimination success via local operations falls short of that achievable globally. It is noteworthy that this ensemble of states constitute an UB -- it spans a 3 dimensional subspace in $\mathbb{C}^2\otimes \mathbb{C}^2$, with only the $\ket{\Psi^-}_{AB}$ orthogonal to it. Nevertheless, as we shall demonstrate, these states are conclusively distinguishable using local operations.
\begin{proposition}\label{propPW}
The set of states $\mathcal{S}^{\text{P-W}}$ can be distinguished conclusively via LOCC.
\end{proposition}
\begin{proof}
The state $\ket{w_1^{\perp}w_2^{\perp}}\equiv\ket{w_1^{\perp}}_A\ket{w_2^{\perp}}_B$ is orthogonal to $\ket{w_1w_1}$ and $\ket{w_2w_2}$, but non-orthogonal to $\ket{w_0w_0}$. Consider the measurement 
\begin{equation*}
\left\{
\begin{array}{l}
\textstyle E_0 ~:= \ket{w_1^{\perp}w_2^{\perp}}\bra{w_1^{\perp}w_2^{\perp}},\textstyle E_{?}~ := \ket{w_1^{\perp}w_2}\bra{w_1^{\perp}w_2} \\
\textstyle E_{??} := \ket{w_1w_2^{\perp}}\bra{w_1w_2^{\perp}},~~
\textstyle E_{???} := \ket{w_1w_2}\bra{w_1w_2}
\end{array}
\right\},
\end{equation*}
which is locally implementable, and ensures local identifying of the state $\ket{w_0}_A \otimes \ket{w_0}_B$ -- clicking of the effect $E_0$ ensures the state was $\ket{w_0}_A \otimes \ket{w_0}_B$, else they register an inconclusive result. Similarly the other two states can also be conclusively identified locally, and thereby ensuring conclusive local distinguishability of the double-trine ensemble.
\end{proof}
At this point, recalling the Lemma \ref{lemmaGUB}, one can conclude that the double trine ensemble is not only a UB, rather it is a GUB. More recently, the local state marking (LSM) task has been introduced as a generalization of local state discrimination (LSD) \cite{LSM}. We recall the formal definition from \cite{LSM}:
\begin{definition}\label{defLSM}
(Sen \textit{et  al.} \cite{LSM}). Given \( m \) states chosen randomly from a known set of mutually orthogonal multipartite quantum states \( \mathcal{S} \equiv \left\{ \ket{\psi_j} \right\}_{j=1}^N \), the \( m \)-LSM task demands correctly answering (or marking),  each of the \( m \) states via LOCC, thus figuring out the permutation of the $m$ states in the process.
\end{definition}
Here, $m$ ranges from 1 to $|\mathcal{S}|$, with the special case $m = 1$ corresponding to the standard local state discrimination (LSD) task. The $m$-LSM task necessarily involves a set of mutually orthogonal states, as the parties \textit{must} correctly determine the exact sequence of the $m$ states. As conclusive state discrimination has already significantly deepened our understanding of local distinguishability, leading to several notable insights, a natural question then arises: if the LSM task is relaxed to allow conclusive outcomes---possibly at the cost of inconclusive ones---what new and interesting consequences might emerge? We shall address this question in this paper, and start by formally defining the conclusive local state marking task (CLSM). 
\begin{definition}\label{def3}
[\(m\)-CLSM] Given \( m \) states chosen randomly from a known set of multipartite quantum states \( \mathcal{S} \equiv \left\{ \ket{\psi_j} \right\}_{j=1}^N \), the \( m \)-CLSM task demands marking each of the \( m \) states via LOCC, conclusively.
\end{definition}
\begin{figure}[t!]
\centering
\includegraphics[width=0.35\textwidth]{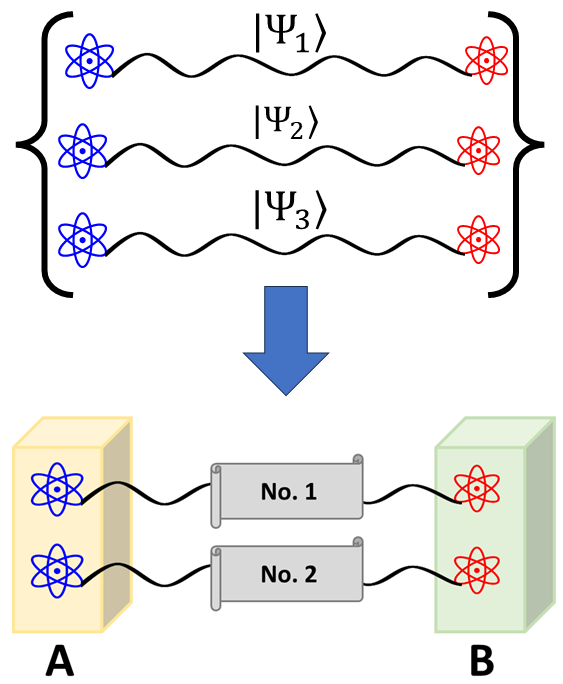}
\caption{(Color online) In the bipartite scenario, the task of $2$-CLSM is demonstrated for a set containing three bipartite states: $Z := \{\ket{\Psi_1}, \ket{\Psi_2}, \ket{\Psi_3}\}$. Two states, chosen randomly from the set, are distributed between spatially separated Alice and Bob without revealing their identities. The possible input pairs are: $\{\ket{\Psi_1\Psi_2}, \ket{\Psi_1\Psi_3}, \ket{\Psi_2\Psi_1}, \ket{\Psi_2\Psi_3}, \ket{\Psi_3\Psi_1}, \ket{\Psi_3\Psi_2}\}$. They must conclusively identify the indices of the two states (No. 1 and No. 2) using LOCC.}
\label{CLSMfig1}
\end{figure}
As in the LSM scenario, here too \( m \) ranges from \( 1 \) to \( |\mathcal{S}| \), with the special case \( m = 1 \) reducing to the standard conclusive local state discrimination (conclusive-LSD) task (see Fig.~\ref{CLSMfig1}). However, unlike the \( m \)-LSM task, in \( m \)-CLSM the parties are only required to provide a correct answer when they choose to respond; otherwise, they may declare an inconclusive outcome by remaining silent. 
At this point, it is important to highlight a crucial distinction between the \( m \)-CLSM task and the multi-copy conclusive local state discrimination (CLSD) task. The latter addresses the problem of distinguishing a state drawn from a known set when multiple copies of the unknown state are provided \cite{CheflesLD}. While linear independence is a necessary condition for conclusive discrimination in the single-copy setting, in the multi-copy setting even linearly dependent ensembles may become conclusively distinguishable, as the tensor product of multiple copies can lead to the desired linear independence. As we will see, no \( m \)-CLSM task can be possible for a set of linearly dependent set of states.

\color{black}
\begin{proposition}\label{propNOTPOSSIBLE}
For any linearly dependent set of multipartite quantum states, the \( m \)-CLSM task is not possible for any \( m \geq 1 \).
\end{proposition}

The detailed proof is discussed in Appendix~\ref{appendixA}, where we have given the proof for $3$-CLSM, which can easily be generalised for general $m$. For clarity, we outline the argument for the case of $3$-CLSM.

For a linearly dependent set of vectors $\{v_i\}^N_{i=1}$, there exist coefficients $\alpha_i$, not all zero, such that
    \(
        \sum_{i=1}^{N} \alpha_i v_i = 0.
    \) 
    Now, we tensor this relation by two fixed vectors $v_{j_2}$ and $v_{j_3}$ from the set. 
    This produces a new equation involving three-fold tensor products of the $v_i$'s. Then, we apply all possible permutations of the three tensor slots, multiply each by the sign of the permutation, 
    and sum them up. This step ensures that only terms with three distinct indices survive, 
    since repeating any state makes the antisymmetric sum vanish.
 After antisymmetrization, all terms with repeated indices vanish. 
    The remaining terms correspond to combinations of three distinct states and form the $3$-CLSM candidates. This provides a nontrivial relation among the $3$-CLSM states and thus they are linearly dependent.

\color{black}

Hence, we now turn our attention to the \(m\)-CLSM task for a linearly independent set of multipartite quantum states. Can the derived set of states constructed with all ordered \(m\)-tuples from the original set become linearly dependent—thereby ruling out the possibility of global conclusive discrimination and, consequently, its local counterpart? As we will demonstrate, this is not the case for the \( m \)-CLSM task.
\begin{proposition}\label{propLI}
For any set of linearly independent multipartite quantum states, the set of states arising in \(m\)-CLSM task is linearly independent.
\end{proposition}
We discuss the proof in Appendix \ref{appendixB}.
The essential idea is that for any set of linearly independent multipartite states, forming $m$-fold tensor products generates a larger set of states that remain linearly independent. The states arising in the $m$-CLSM task are simply a subset of these tensor-product states. Since any subset of a linearly independent set is itself linearly independent, the $m$-CLSM states inherit this property.

\color{black}

Since the \(m\)-CLSM task can be viewed as a CLSD task in higher dimensional systems, given a set of linearly independent set of multipartite states \(\mathcal{S}\) a necessary and sufficient condition for its \(m\)-CLSM can be obtained by generalizing the Lemma \ref{lemmaCheflesProductDetecting}. Instead of stating this, in the following we rather establish some interesting connections between CLSD and CLSM.

\begin{theorem}\label{theo1}
For a set of multipartite states $\mathcal{S}$, conclusive local state discrimination (i.e. $1$-CLSM) always implies m-conclusive local state marking (i.e. $m$-CLSM ~~~$\forall m \in \{1,2,\cdots ,|\mathcal{S}|\}$).
\end{theorem}

The detailed proof is provided in Appendix \ref{appendixC}. For a set of states $\mathcal{S}$ that can be conclusively distinguished locally, the $m$-CLSM task for this set can be viewed as a standard CLSD task of higher-dimensional linearly independent states. Since every member of the original set of states is conclusively locally identifiable, this fact also holds for the $m$-CLSM states. This is the main idea behind the proof.\\
\color{black}

As the ensemble \(\mathcal{S}^{\text{P-W}}\) in Eq.(\ref{eqPW}) allows CLSD (see Proposition \ref{propPW}), thus according to Theorem \ref{theo1} it also allows any $m$-CLSM. We thus proceed to provide a different example of product states ensemble that does not allow CLSM. For that first consider the symmetric informationally complete (SIC) ensemble of the qubit:
\begin{align}\label{SIC}
\mathrm{SIC}
\equiv \left\{\begin{aligned}
&~~|s_1\rangle:= |0\rangle,~\text{and ~for}~j\in\{2,3,4\}\\
&|s_j\rangle:= \frac{1}{\sqrt{3}} \left( |0\rangle + e^{2\pi\iota (j-2)/3}\sqrt{2} |1\rangle \right)
\end{aligned} \right\}.   
\end{align}
Now, the double-SIC ensemble
\begin{align}\label{SIC}
\mathcal{S}^{\uparrow\hspace{-.05cm}\uparrow}:=\left\{\ket{s_i}_A\otimes\ket{s_i}_B~\text{s.t.}~\ket{s_i}\in\mathrm{SIC,~~i\in\{1,\cdots,4}
\right\}\},  
\end{align}
might be considered as a natural generalization of the Peres-Wootters ensemble \(\mathcal{S}^{\text{P-W}}\), in terms of symmetry of each party's qubit subsystem. However, the states in \(\mathcal{S}^{\text{P-W}}\) are linearly independent, whereas the ensemble \(\mathcal{S}^{\uparrow\hspace{-.05cm}\uparrow}\) spans a three dimensional subspace of \(\mathbb{C}^2\otimes\mathbb{C}^2\), in particular the subspace orthogonal to the state \(\ket{\Psi^-}\). Hence, the states in \(\mathcal{S}^{\uparrow\hspace{-.05cm}\uparrow}\) becomes linearly dependent making its conclusive discrimination impossible (see Lemma \ref{lemmaCheflesGlobal}). While the spins in \(\mathcal{S}^{\uparrow\hspace{-.05cm}\uparrow}\) are in parallel configuration, we thus consider another ensemble, namely the anti-parallel double-SIC ensemble
\begin{align}\label{anti-SIC}
\mathcal{S}^{\uparrow\hspace{-.05cm}\downarrow}:=\left\{
\ket{\gamma_i}:=\ket{s_i}_A\otimes\ket{s^\perp_i}_B~\text{s.t.}~\ket{s_i}\in\mathrm{SIC}
\right\},  
\end{align}
where \(\ket{\psi^\perp}\) denotes the state orthogonal to \(\ket{\psi}\). Our next result establishes a nonlocal behavior of this ensemble.  
\begin{proposition}\label{propSICASIC}
The ensemble of states \(\mathcal{S}^{\uparrow\hspace{-.05cm}\downarrow}\) does not allow conclusive local discrimination, while it allows conclusive global discrimination.     
\end{proposition}
\begin{proof}
The states in \(\mathcal{S}^{\uparrow\hspace{-.05cm}\downarrow}\) are linearly independent and thus allow conclusive global discrimination (Lemma \ref{lemmaCheflesGlobal}). This set trivially forms a UPB as the states span the full two-qubit Hilbert space. Moreover, it can be trivially shown that the vector orthogonal to the span of \(\mathcal{S}^{\uparrow\hspace{-.05cm}\downarrow}\setminus\{\ket{\gamma_i}\}\) is entangled for all \(i\in\{1,\cdots,4\}\). Therefore, \(\mathcal{S}^{\uparrow\hspace{-.05cm}\downarrow}\) is not a  GUPB, and according to Lemma \ref{lemmaGUB} it does not allow conclusive local discrimination.  
\end{proof}

Proposition \ref{propSICASIC} thus establishes a stronger nonlocality of the ensemble \(\mathcal{S}^{\uparrow\hspace{-.05cm}\downarrow}\) as compared to the ensemble \(\mathcal{S}^{\text{P-W}}\). In this direction, the product states used to demonstrate nonlocality in the work of Peres and Wootters~\cite{PeresWooters} (see Proposition~\ref{propPW}), as well as those introduced by Bennett {\it et al.}, are conclusively distinguishable via LOCC. Consequently, by Theorem~\ref{theo1}, any $m$-CLSM protocol is possible for these sets. In contrast, the product states comprising the set $\mathcal{S}^{\uparrow\hspace{-.05cm}\downarrow}$ are not locally conclusively distinguishable. A similar property holds for the set $\mathcal{S}^{\text{D}}$, introduced by Duan {\it  et al.}, which is also locally conclusively indistinguishable. As we will show later, the converse of Theorem~\ref{theo1} does not hold (See Appendix  \ref{appendixF}. We shall demonstrate this by providing an example of a set of two mutually 
orthogonal states for which CLSD is not possible. Interestingly, even under 
arbitrarily large but finite multiple-copy assistance, the set remains 
conclusively indistinguishable. However, we find that its CLSM --- in this 
case, $2$-CLSM --- is indeed possible.

  The implications of converse of Theorem \ref{theo1} not holding are particularly interesting and constitute one of the main motivations of the present work. For a set of states $\mathcal{S}$, the existence of an integer $m^\star \in (1, |\mathcal{S}|]$ for which $m^\star$-CLSM is not possible indicates a stronger form of nonlocality than what has previously been identified in the literature. This phenomenon is precisely what we establish for the set $\mathcal{S}^{\uparrow\hspace{-.05cm}\downarrow}$  in the following theorem. 
\begin{theorem}\label{theo2}
The ensemble of states \(\mathcal{S}^{\uparrow\hspace{-.05cm}\downarrow}\) does not allow \(2\)-CLSM.     
\end{theorem}

The key steps involve demonstrating that the set \(\mathcal{S}^{\uparrow\hspace{-.05cm}\downarrow}_{(2)}\), defined as the $2$-CLSM of the ensemble \(\mathcal{S}^{\uparrow\hspace{-.05cm}\downarrow}\), constitutes an unextendible product basis (UPB), but not a genuine
one. Consequently, by invoking Lemma \ref{lemmaGUB}, the claim follows. The proof is discussed in detail in Appendix \ref{appendixD}. Interestingly, \(\mathcal{S}^{D}\) also exhibits the same phenomenon, and we discuss it in Appendix \ref{appendixE}.\\
\color{black}

{\it Discussion.}\label{conclusion}
In this work, we unify two key concepts in state discrimination: conclusive local state discrimination (CLSD) and local state marking (LSM), by introducing a new class of tasks called \(m\)-conclusive local state marking (\(m\)-CLSM). This framework highlights subtle aspects of quantum nonlocality. Propositions~\ref{propNOTPOSSIBLE} and~\ref{propLI} reveal structural features of \(m\)-CLSM, while Theorem~\ref{theo2} (also see  Appendix \ref{appendixE}) shows that certain product states exhibit stronger nonlocality than those in earlier constructions by Bennett \textit{et al.}~\cite{BennettNLWE} and Peres--Wootters~\cite{PeresWooters}.

These results open up promising directions for future work. Our use of adaptive LOCC protocols, though operationally feasible, yields low success probabilities. Exploring whether global measurements could improve this is a natural next step.  A key open question is whether one can construct a binary ensemble of mutually orthogonal quantum states that are not only conclusively indistinguishable via LOCC, but also disallow CLSM (here $2$-CLSM). In order to construct such an example, the two 
subspaces corresponding to the supports of the two states, say 
$\rho_{AB}$ and $\sigma_{AB}$ with 
$\mathrm{supp}(\rho_{AB}) \perp \mathrm{supp}(\sigma_{AB})$, 
must themselves be completely entangled. Under their $2$-CLSM, which 
involves tensoring two distinct states, the supports corresponding to 
$\rho_{A_1B_1} \otimes \sigma_{A_2B_2}$ and $\sigma_{A_1B_1} \otimes \rho_{A_2B_2}$ 
would also need to be completely entangled across the $A_1A_2 \,|\, B_1B_2$ bipartition. We believe these works \cite{duansup,cubitt} might be helpful in this construction. Thus, our work provides a direct connection to the vastly explored area of superactivation phenomenon of \textit{zero error capacity of quantum channels} \cite{duansup,cubitt}  and Unextendible Bases (UB) \cite{DiVincenzo2003} and Completely Entangled Subspaces (CES) \cite{BHAT}.

Additionally, comparing different \(m\)-CLSM tasks remains an open challenge; for example, finding state sets where \(2\)-CLSM fails but \(3\)-CLSM succeeds would be particularly insightful.

{\bf Acknowledgment:} It is a pleasure to thank Manik Banik for various
stimulating discussions and useful suggestions.

\onecolumngrid
\appendix

\section{Proof of Proposition \ref{propNOTPOSSIBLE}}\label{appendixA}
   \setcounter{proposition}{1}  
\begin{proposition}\label{lemmaNOTPOSSIBLE}
For any linearly dependent set of multipartite quantum states, the \( m \)-CLSM task is not possible for any \( m \geq 1 \).
\end{proposition}
\begin{proof}
For simplicity, we give the proof for $3$-CLSM. For general $m$, it can be easily generalised. \\
Let us suppose that $\{v_1,\dots,v_N\}$ is a set of distinct vectors $(N>2)$ with a linear dependence
\begin{equation}
\sum_{i=1}^N \alpha_i v_i = 0, 
\quad \text{not all } \alpha_i = 0.
\label{eq:dep}
\end{equation}

Now, we fix two indices $j_2,j_3$ distinct from each other. Tensoring Eq.~(\ref{eq:dep}) with $v_{j_2}\otimes v_{j_3}$ yields
\begin{equation}
\sum_{i=1}^N \alpha_i \, v_i \otimes v_{j_2} \otimes v_{j_3} = 0.
\label{eq:tensor}
\end{equation}

We apply every permutation $\sigma \in S_3$ (the symmetric group on three slots) to Eq.~(\ref{eq:tensor}), multiply by $\mathrm{sgn}(\sigma)$, and sum over all $\sigma$. This antisymmetrization produces
\begin{equation}
\sum_{i=1}^N \alpha_i 
\Bigg( \sum_{\sigma \in S_3} \mathrm{sgn}(\sigma)\,
\sigma \cdot (v_i \otimes v_{j_2} \otimes v_{j_3}) \Bigg) = 0.
\label{eq:antisym}
\end{equation}

Now, define for each $i$
\begin{equation}
A_i := \sum_{\sigma \in S_3} \mathrm{sgn}(\sigma)\,
\sigma \cdot (v_i \otimes v_{j_2} \otimes v_{j_3}).
\end{equation}
Then Eq.~(\ref{eq:antisym}) reads
as  $\sum_{i=1}^N \alpha_i A_i = 0.$\\

Hence, two cases arise:
\begin{itemize}
    \item If $i \in \{j_2,j_3\}$, then some tensor slots are repeated. In this case, $A_i=0$ because exchanging the two identical slots yields cancellation.
    \item If $i \notin \{j_2,j_3\}$, then $A_i$ is a nonzero alternating sum of tensors with three distinct indices $(i,j_2,j_3)$ in all orders.
\end{itemize}

Thus the sum simplifies to
\begin{equation}
\sum_{i \notin \{j_2,j_3\}} \alpha_i A_i = 0.
\end{equation}
Since not all $\alpha_i$ vanish and only two indices are excluded, at least one surviving coefficient contributes. This provides a nontrivial linear relation among the $3$-CLSM states and thus they are linearly dependent.
\end{proof}
\color{black}

\section{Proof of Proposition \ref{propLI}}\label{appendixB}
\begin{proposition}
    For any set of linearly independent multipartite quantum states, the set of states arising in \(m\)-CLSM task is linearly independent.
\end{proposition}

\begin{proof}

    Let \( \mathcal{S} \equiv \left\{\ket{\psi_1},\cdots,\ket{\psi_N}\right\} \subset \bigotimes_{i=1}^K \mathbb{C}^{d_i}_{A_i}\) be a linearly independent set of states. Span of \(\mathcal{S}\) constitute a \(N\)-dimensional subspace, i.e., \(\text{Span}(\mathcal{S}):=\mathcal{H}\subset\bigotimes_{i=1}^K \mathbb{C}^{d_i}_{A_i}\). Consider now a new set of state \(\mathcal{T}_m\) constituted by \(m\)-fold tenor product of the states from \(\mathcal{S}\), i.e. \(\mathcal{T}_m\equiv\{\ket{\psi_{i_1}} \otimes\cdots\otimes\ket{\psi_{i_m}}\}^{N}_{i_1,\cdots,~i_m=1}\). The set of states is \(\mathcal{T}_m\) are again linearly independent and they span a \(N^m\)-dimensional vector space \(\mathcal{H}^{\otimes m}\) \cite{Steven_LinearAlgebra}. The \(m\)-CLSM task constituted for the set \(\mathcal{S}\) consists of states \(\mathcal{S}^\prime_m\equiv\{\ket{\psi_{i_1}} \otimes\cdots\otimes\ket{\psi_{i_m}}\}^{N}_{i_1,\cdots,i_m=1}\), that for a subspace of the set \(\mathcal{T}_m\).  Since any subset of a linearly independent set of states is again linearly independent \cite{Steven_LinearAlgebra}, therefore the set of states \( \mathcal{S}^\prime_m\) are linearly independent. This completes the proof.
\end{proof}

\section{Proof of Theorem \ref{theo1}}\label{appendixC}
\setcounter{theorem}{0}  
\begin{theorem}
    For a set of multipartite states $\mathcal{S}$, conclusive local state discrimination (i.e. $1$-CLSM) always implies m-conclusive local state marking (i.e. $m$-CLSM ~~~$\forall m \in \{1,2,\cdots ,|\mathcal{S}|\}$).
\end{theorem}
\begin{proof}

Let \( \mathcal{S}_N \equiv \left\{\ket{\psi_1},\cdots,\ket{\psi_N}\right\} \subset \bigotimes_{i=1}^K \mathbb{C}^{d_i}_{A_i}:= \mathcal{H} \) be a set of states that allows conclusive local discrimination. The problem of $m$-CLSM for the set $\mathcal{S}_N$ can be reformulated as a CLSD problem of the set of states $\mathcal{S}^{[m]}_{\mathcal{P}[\{N\}]}\equiv\left\{\mathcal{P}^{[m]}\left(\otimes_{i=1}^N\ket{\psi_i}\right)\right\}\subset\mathcal{H}^{\otimes m}$, where $\left\{\mathcal{P}^{[m]}\left(\otimes_{i=1}^N\ket{\psi_i}\right)\right\}$ denotes the set of tensor product states generated through permutations of $m$ distinct indices from the set $\{1,\cdots,N\}$. For instance, for the simple case where $N=3, m=2$, it means
\begin{align}
~~~~~~~\mathcal{S}^{[2]}_{\mathcal{P}[\{3\}]} 
&:= \left\{ \mathcal{P} \left( \otimes_{i=1}^2 \ket{\psi_i} \right) \right\} \notag \\
&= \left\{ 
\ket{\psi_{i_1}} \otimes \ket{\psi_{i_2}} \;\middle|\;
i_1, i_2 \in \{1,2,3\},~ i_1 \neq i_2 
\right\} \notag \\
&= \left\{
\ket{\psi_1\psi_2},\,
\ket{\psi_1\psi_3},\,
\ket{\psi_2\psi_1},\,
\ket{\psi_2\psi_3}, \right. \notag \\
&\quad \left.
\ket{\psi_3\psi_1},\,
\ket{\psi_3\psi_2}
\right\}, \text{where } \ket{xy} := \ket{x} \otimes \ket{y} \notag
\end{align}

 The states in $\mathcal{S}^{[m]}_{\mathcal{P}[\{N\}]}$ can be expressed group-wise as follows,
\begin{align*}
\mathcal{G}^{[m]}_l:=\ket{\psi_l}\otimes\mathcal{S}^{[m-1]}_{\mathcal{P}[\{N\}\setminus l]}\equiv\ket{\psi_l}\otimes\left\{\mathcal{P}\left(\otimes_{i\neq l}^{(m-1)}\ket{\psi_i}\right)\right\},
\end{align*}
where $l\in\{1,\cdots,N\}$. Now, the groups $\mathcal{G}^{[m]}_l$ make disjoint partitions of the set $\mathcal{S}^{[m]}_{\mathcal{P}[\{K\}]}$, {\it i.e.,} $\mathcal{S}^{[m]}_{\mathcal{P}[\{K\}]}\equiv \bigcup_{l=1}^K \mathcal{G}^{[m]}_l$ s.t. $\mathcal{G}^{[m]}_l\cap\mathcal{G}^{[m]}_{l^\prime}=\emptyset$ whenever $l\neq l^\prime$. Since the states in $\mathcal{S}_N$ are locally conclusively distinguishable, by local operations on the first part of the tensor product states in $\mathcal{S}^{[m]}_{\mathcal{P}[\{N\}]}$ we can know conclusively in which of the above groups the given state lies. When a conclusive outcome clicks, the group turns out to be $\mathcal{G}^{[m]}_{l^\star}$ (i.e., if the index $l$ has been identified to be $l^*$), the given state $\ket{\psi_{l^\star}}\otimes(\cdots)$ evolves to $\ket{\psi^\prime_{l^\star}}\otimes(\cdots)$ due to the LOCC protocol, where the term within the brackets remain unchanged and hence further LOCC protocols can be applied on them. The group of states $\mathcal{G}^{[m]}_{l^\star}=\ket{\psi^\prime_{l^\star}}\otimes\mathcal{S}^{[m-1]}_{\mathcal{P}[\{K\}\setminus {l^\star}]}$ can be further partitioned into disjoint subsets as,
\begin{align*}
\mathcal{G}^{[m]}_{l^\star}\equiv\bigcup\mathcal{G}^{[m]}_{l^\star,u}~~\mbox{s.t.}~~\mathcal{G}_{l^\star,u}\cap\mathcal{G}_{l^\star,u^\prime}=\emptyset~\forall~u\neq u^\prime\nonumber,\\
\mbox{where}~\mathcal{G}^{[m]}_{l^\star,u}:=\ket{\psi^\prime_{l^\star}}\otimes\ket{\psi_{u}}\otimes\mathcal{S}^{[m-2]}_{\mathcal{P}[\{N\}\setminus \{l^\star,u\}]}
\end{align*} 
and $u,u^\prime\in\{1,\cdots,N\}\setminus l^\star$. Since any subset of a set of states that is locally conclusively distinguishable remains locally conclusively distinguishable, the index \( u \) can be identified conclusively by applying an appropriate local protocol to the \( \ket{\psi_{u}} \) portion of the given state. As in the previous case, the other parts of the state remain unchanged. This process can be iterated until the state in \( \mathcal{S}^{[m]}_{\mathcal{P}[\{N\}]} \) is conclusively identified, which in turn allows for the m-conclusive marking of the state in \( \mathcal{S}^{[m]}_{N} \). This completes the proof.
\end{proof}

\section{ Proof of Theorem \ref{theo2}}\label{appendixD}
\begin{theorem}
    The ensemble of states \(\mathcal{S}^{\uparrow\hspace{-.05cm}\downarrow}\) does not allow \(2\)-CLSM. 
\end{theorem}
\begin{proof}
The \(2\)-CLSM of the ensemble \(\mathcal{S}^{\uparrow\hspace{-.05cm}\downarrow}\) boils down to CLSD of the ensemble 
\begin{align*}
\mathcal{S}^{\uparrow\hspace{-.05cm}\downarrow}_{(2)}\equiv\left\{\begin{aligned}
&\ket{\chi_{ij}}:=\ket{s_i}\ket{s_j}\otimes\ket{s^\perp_i}\ket{s^\perp_j},\\
&\text{s.t.}~\ket{s_i}\in\mathrm{SIC}~\&~i\neq j
\end{aligned}\right\} \subset\mathbb{C}^4_A\otimes\mathbb{C}^4_B. 
\end{align*}

The key steps involve demonstrating that the set \(\mathcal{S}^{\uparrow\hspace{-.05cm}\downarrow}_{(2)}\) constitutes an unextendible product basis (UPB), but not a genuine
one. Consequently, by invoking Lemma  \textcolor{red}{3}, the claim follows.\\

Clearly we have \(\mathrm{D}[\mathrm{Spn}\{\mathcal{S}^{\uparrow\hspace{-.05cm}\downarrow}_{(2)}\}]=12\), where \(\mathrm{Spn}\{\star\}\) denotes span of a set \(\{\star\}\) and \(\mathrm{D}[\mathrm{Spn}\{\star\}]\) denotes its dimension. As it turn out the ensemble \(\mathcal{S}^{\uparrow\hspace{-.05cm}\downarrow}_{(2)}\) forms a UPB. To prove this let us consider a product state $\ket{\phi} = \ket{\alpha} \otimes \ket{\beta} \in \mathbb{C}^4 \otimes \mathbb{C}^4$ that lies in the orthogonal complement of \(\mathrm{Spn}\{\mathcal{S}^{\uparrow\hspace{-.05cm}\downarrow}_{(2)}\}\). Therefore we have  
\begin{align}
\langle s_i \otimes s_j | \alpha \rangle =0~~\text{and/or}~~\langle s_i^\perp \otimes s_j^\perp | \beta \rangle = 0,\nonumber\\
\forall~i,j\in\{1,\cdots4\},~~\text{where}~~i\neq j.
\end{align}
Consider a disjoint partition of the ensemble \(\mathcal{S}^{\uparrow\hspace{-.05cm}\downarrow}_{(2)}\) i.e.
\begin{align}
\mathcal{S}_A\cup\mathcal{S}_B=\mathcal{S}^{\uparrow\hspace{-.05cm}\downarrow}_{(2)}~~\text{s.t.}~~\mathcal{S}_A\cap\mathcal{S}_B=\emptyset,
\end{align}
and let, \(r_A := \text{Rank}\left\{\ket{\chi_{ij}} \in \mathcal{S}_A \right\}\) be the local rank of subset \(\mathcal{S}_A\) as seen by Alice, and similarly \(r_B := \text{Rank}\left\{ \ket{\chi_{ij}} \in \mathcal{S}_B \right\}\) be the local rank of subset \(\mathcal{S}_B\) as seen by Bob. Now, the ensemble \(\mathcal{S}^{\uparrow\hspace{-.05cm}\downarrow}_{(2)}\) is extendible if and only if there exists a partition of it such that both \(r_A < 4\) and \(r_B < 4\) \cite{Bennett99(1)}. A straightforward but lengthy calculation shows that any partition containing seven or more states in \(\mathcal{S}_A\) has local rank \(r_A=4\), likewise for \(\mathcal{S}_B\). Thus, to find the desired product state, we must divide \(\mathcal{S}^{\uparrow\hspace{-.05cm}\downarrow}_{(2)}\) into two sets of equal cardinalities, {\it i.e} six. We have only the following four partitionings where \(\mathcal{S}_A\) has local rank \(r_A=3\): 
\begin{subequations}
\begin{align}
\{\ket{\chi_{12}},\ket{\chi_{13}},\ket{\chi_{14}},\ket{\chi_{21}},\ket{\chi_{31}},\ket{\chi_{41}}\},\nonumber\\
\{\ket{\chi_{21}},\ket{\chi_{23}},\ket{\chi_{24}},\ket{\chi_{12}},\ket{\chi_{32}},\ket{\chi_{42}}\},\nonumber\\
\{\ket{\chi_{31}},\ket{\chi_{32}},\ket{\chi_{34}},\ket{\chi_{13}},\ket{\chi_{23}},\ket{\chi_{43}}\},\nonumber\\
\{\ket{\chi_{41}},\ket{\chi_{42}},\ket{\chi_{43}},\ket{\chi_{14}},\ket{\chi_{24}},\ket{\chi_{34}}\}\nonumber.
\end{align}  
\end{subequations}
However, in all of these cases the local rank \(r_B\) of \(\mathcal{S}_B\) turns out to be \(4\). One can also consider other partitions having less number of states in \(\mathcal{S}_A\) with \(r_A\le3\), but in theses cases too we have \(r_B=4\). This ensures that no product state lies in the orthogonal complement of \(\mathrm{Spn}\{\mathcal{S}^{\uparrow\hspace{-0.05cm}\downarrow}_{(2)}\}\), thereby confirming it to be an UPB. 

To explicitly demonstrate the non-genuineness of the UPB, we identify and construct a proper subset of \(\mathcal{S}^{\uparrow\hspace{-.05cm}\downarrow}_{(2)}\) that independently forms a UPB. 
Consider the proper subset $\mathcal{T} := \mathcal{S}^{\uparrow\hspace{-.05cm}\downarrow}_{(2)} \setminus \{ \ket{\chi_{12}} \}$. Regarding the ensemble $\mathcal{T}$  few observations are in order:
\begin{itemize}
    \item  We have only the following five-element subset 
 partitionings where \(\mathcal{S}_A\) has local rank \(r_A=3\): 

\begin{subequations}
\begin{align*}
\begin{array}{ll}
\{\ket{\chi_{13}}, \ket{\chi_{14}}, \ket{\chi_{21}}, \ket{\chi_{31}}, \ket{\chi_{41}}\}, &
\{\ket{\chi_{13}}, \ket{\chi_{21}}, \ket{\chi_{23}}, \ket{\chi_{24}}, \ket{\chi_{43}}\}, \\[5pt]
\{\ket{\chi_{13}}, \ket{\chi_{23}}, \ket{\chi_{31}}, \ket{\chi_{32}}, \ket{\chi_{34}}\}, &
\{\ket{\chi_{13}}, \ket{\chi_{23}}, \ket{\chi_{31}}, \ket{\chi_{32}}, \ket{\chi_{43}}\}, \\[5pt]
\{\ket{\chi_{13}}, \ket{\chi_{23}}, \ket{\chi_{31}}, \ket{\chi_{34}}, \ket{\chi_{43}}\}, &
\{\ket{\chi_{13}}, \ket{\chi_{23}}, \ket{\chi_{32}}, \ket{\chi_{34}}, \ket{\chi_{43}}\}, \\[5pt]
\{\ket{\chi_{13}}, \ket{\chi_{23}}, \ket{\chi_{41}}, \ket{\chi_{42}}, \ket{\chi_{43}}\}, &
\{\ket{\chi_{13}}, \ket{\chi_{31}}, \ket{\chi_{32}}, \ket{\chi_{34}}, \ket{\chi_{43}}\}, \\[5pt]
\{\ket{\chi_{14}}, \ket{\chi_{21}}, \ket{\chi_{23}}, \ket{\chi_{24}}, \ket{\chi_{34}}\}, &
\{\ket{\chi_{14}}, \ket{\chi_{24}}, \ket{\chi_{31}}, \ket{\chi_{32}}, \ket{\chi_{34}}\}, \\[5pt]
\{\ket{\chi_{14}}, \ket{\chi_{24}}, \ket{\chi_{34}}, \ket{\chi_{41}}, \ket{\chi_{42}}\}, &
\{\ket{\chi_{14}}, \ket{\chi_{24}}, \ket{\chi_{34}}, \ket{\chi_{41}}, \ket{\chi_{43}}\}, \\[5pt]
\{\ket{\chi_{14}}, \ket{\chi_{24}}, \ket{\chi_{34}}, \ket{\chi_{42}}, \ket{\chi_{43}}\}, &
\{\ket{\chi_{14}}, \ket{\chi_{24}}, \ket{\chi_{41}}, \ket{\chi_{42}}, \ket{\chi_{43}}\}, \\[5pt]
\{\ket{\chi_{14}}, \ket{\chi_{34}}, \ket{\chi_{41}}, \ket{\chi_{42}}, \ket{\chi_{43}}\}, &
\{\ket{\chi_{21}}, \ket{\chi_{23}}, \ket{\chi_{24}}, \ket{\chi_{31}}, \ket{\chi_{41}}\}, \\[5pt]
\{\ket{\chi_{21}}, \ket{\chi_{23}}, \ket{\chi_{24}}, \ket{\chi_{32}}, \ket{\chi_{42}}\}, &
\{\ket{\chi_{21}}, \ket{\chi_{31}}, \ket{\chi_{32}}, \ket{\chi_{34}}, \ket{\chi_{41}}\}, \\[5pt]
\{\ket{\chi_{21}}, \ket{\chi_{31}}, \ket{\chi_{41}}, \ket{\chi_{42}}, \ket{\chi_{43}}\}, &
\{\ket{\chi_{23}}, \ket{\chi_{31}}, \ket{\chi_{32}}, \ket{\chi_{34}}, \ket{\chi_{43}}\}, \\[5pt]
\multicolumn{2}{c}{
\{\ket{\chi_{24}}, \ket{\chi_{34}}, \ket{\chi_{41}}, \ket{\chi_{42}}, \ket{\chi_{43}}\}
}
\end{array}
\end{align*}
\end{subequations}
However, the local rank \(r_B=4\).
\item  We have only the following six-element partitionings where \(\mathcal{S}_A\) has local rank \(r_A=3\):
\begin{subequations}
\begin{align*}
\begin{array}{ll}
\{\ket{\chi_{13}}, \ket{\chi_{23}}, \ket{\chi_{31}}, \ket{\chi_{32}}, \ket{\chi_{34}}, \ket{\chi_{43}}\}, &
\{\ket{\chi_{14}}, \ket{\chi_{24}}, \ket{\chi_{34}}, \ket{\chi_{41}}, \ket{\chi_{42}}, \ket{\chi_{43}}\}
\end{array}
\end{align*}
\end{subequations}
Again, the local rank \(r_B=4\).
\item One can also consider other partitions having less number of states in \(\mathcal{S}_A\) with \(r_A\le3\), but in theses cases too we have \(r_B=4\).

\end{itemize}

Thus, in all of these cases the local rank \(r_B\) of \(\mathcal{S}_B\) turns out to be \(4\). This ensures that \(\mathcal{S}^{\uparrow\hspace{-0.05cm}\downarrow}_{(2)}\), is not a GUPB and thus is not conclusively distinguishable via LOCC due to Lemma \textcolor{red}{3}.
\end{proof}

\color{black}

  \begin{figure}[h!]\label{CLSMfig2}
	\centering
	\includegraphics[width=0.6
    \textwidth]{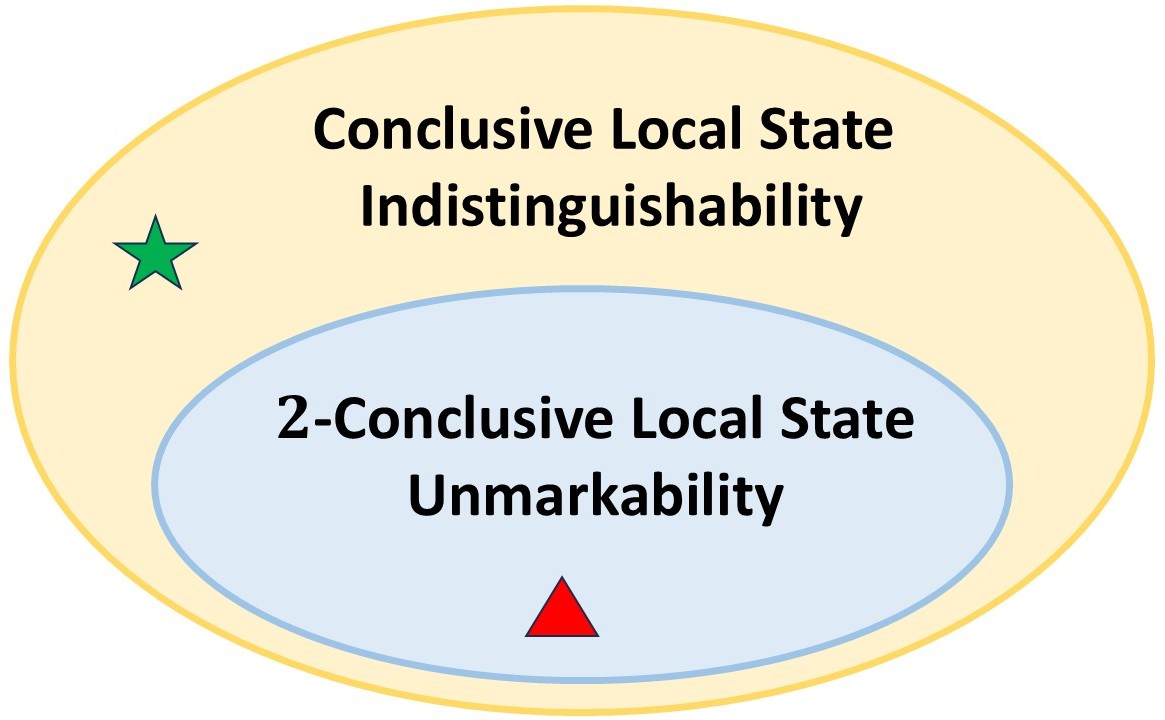}
	\caption{(Color online) A strict hierarchy between the tasks of CLSD and $2$-CLSM is shown. Theorem~\ref{theo4} ( also Proposition \ref{propositionUPB}) is represented by the symbol $\textcolor{green!50!black}{\bigstar}$, which denotes a binary ensemble of states whose CLSD is not possible but for which $2$-CLSM, {\it i.e.} CLSM is possible. In contrast, Proposition \ref{propositionBell} corresponds to the symbol $\textcolor{red}{\scalebox{1.5}{$\blacktriangle$}}$, representing a set of states that is not locally conclusively distinguishable, and for which even $2$-CLSM is not possible.}
	\end{figure}

\section{ The ensemble of states \(\mathcal{S}^{D}\) does not allow \(2\)-CLSM. }\label{appendixE}

\setcounter{theorem}{2} 
\begin{theorem}\label{theo3}
The set of states \(\mathcal{S}^{D}\) does not allow \(2\)-CLSM.
\end{theorem}

\begin{proof}

The set 

\begin{align*}\label{Duanproductstates}
\mathcal{S}^{\text{D}}
\equiv \left\{\begin{aligned}
 &\ket{D_1}:=\ket{0}_A \ket{0}_B,~~~ \ket{D_2}:= \ket{1}_A \ket{1}_B\\
&\ket{D_3}:= \ket{+}_A  \ket{+}_B,~ \ket{D_4}:=\ket{i_+}_A  \ket{i_-}_B
\end{aligned} \right\},    
\end{align*} where $\ket{\pm}:=1/\sqrt{2}(\ket{0} \pm\ket{1})$, $\ket{i_{\pm}}:=1/\sqrt{2}(\ket{0} \pm\iota\ket{1})$, with $\iota:=\sqrt{-1}$\\

can be equivalently written as 
\begin{align*}
\mathcal{S}^{\text{D}}
\equiv \left\{\begin{aligned}
 &\ket{D_1}:=\ket{d^1}_A \ket{d^1}_B,~ \ket{D_2}:= \ket{d^2}_A \ket{d^2}_B\\
&\ket{D_3}:= \ket{d^3}_A  \ket{d^3}_B,~ \ket{D_4}:=\ket{d^4}_A  \ket{d^4}_B
\end{aligned} \right\},    
\end{align*}
with\[
\left\{
\begin{array}{ll}
\ket{d^1}_A := \ket{0}, & \ket{d^1}_B := \ket{0} \\
\ket{d^2}_A := \ket{1}, & \ket{d^2}_B := \ket{1} \\
\ket{d^3}_A := \ket{+}, & \ket{d^3}_B := \ket{+} \\
\ket{d^4}_A := \ket{i_+}, & \ket{d^4}_B := \ket{i_-}
\end{array}
\right\}.
\]
where $\ket{\pm}:=1/\sqrt{2}(\ket{0} \pm\ket{1})$, $\ket{i_{\pm}}:=1/\sqrt{2}(\ket{0} \pm\iota\ket{1})$, with $\iota:=\sqrt{-1}$.

From Def. \textcolor{red}{1}, we need each member of the $2$-CLSM set to be conclusively identifiable. The $2$-CLSM of the set \(\mathcal{S}^{D}\) is:
\begin{align*}
\mathcal{S}^{\text{D}}_{(2)} \equiv \left\{
\begin{array}{l}
\ket{\xi_{ij}} := 
\ket{d^i}_{A_1} \ket{d^j}_{A_2} \otimes 
\ket{d^i}_{B_1} \ket{d^j}_{B_2}, \\[0.5em]
\text{where}~~i\neq j ~~ \forall~i,j\in\{1,\cdots,4\},~~
\end{array}
\right\}
\end{align*}

Thus, for each and every member of $\mathcal{S}^{\text{D}}_{(2)}: \ket{\xi_{i^{\prime}j^{\prime}}}$, there should exist a product state \( |\phi_{ij}\rangle \) such that $\langle \phi_{ij} \vert \xi_{i'j'} \rangle = 0$ for \( (i,j) \neq (i^{\prime},j^{\prime}) \) and \( \langle \phi_{ij} |  \xi_{ij} \rangle \neq 0 \)  We show, on the contrary, that it is not possible, using Lemma \textcolor{red}{2}, where  $i \neq j$ and $i^{\prime}\neq j^{\prime}$.  The above states in $\mathcal{S}^{\text{D}}_{(2)}$ are linearly independent, and hence, they are conclusively distinguishable using global measurements. We now investigate the status of conclusive local state discrimination for this set of states by trying to find a product state that can detect $\ket{\xi_{12}}_{AB}$.  Assume there exists a product state $\ket{\chi} \in \mathbb{C}^4 \otimes \mathbb{C}^4$, such that \( \ket{\chi} = \ket{\alpha}_{A} \otimes \ket{\beta}_{B} \), where:
\begin{align*}
\ket{\alpha} &= (a_1\ket{00}+a_2\ket{01}+a_3\ket{10}+a_4\ket{11})_{A_1A_2} \quad\\
\ket{\beta}  &= (b_1\ket{00}+b_2\ket{01}+b_3\ket{10}+b_4\ket{11})_{B_1B_2} \quad 
\end{align*}
, where $\{\ket{i}_A\}$ and $\{\ket{j}_B\}$ respectively denote the computational basis set for subsystems of Alice's and Bob's, respectively, along with  $\sum |a_i|^2 = 1$ and $ \sum |b_i|^2 = 1$. We require that \( \ket{\chi} \) is orthogonal to all \( \ket{\xi_{ij}} \)  except \( \ket{\xi_{12}} \). However, as we shall see, such a consistent solution cannot be found.\\    
Let us consider the $2$-CLSM of this set. Writing in \( A_1 A_2 | B_1 B_2 \) format, with $A\equiv A_1A_2$, likewise for Bob's subsystem, let us look at the 12 states that the two parties Alice and Bob need to conclusively distinguish via LOCC---
\begin{equation} \label{duan}
\left\{
\begin{aligned}
\ket{\xi_{12}} &= \ket{01}_{A}  \ket{01}_B,  & ~~ \ket{\xi_{13}} &= \ket{0+}_A  \ket{0+}_B, & ~~ \ket{\xi_{14}} &= \ket{0i_+}_A \ket{i-}_B,  & ~~ \ket{\xi_{21}} &= \ket{10}_A \ket{10}_B,  \\
\ket{\xi_{23}} &= \ket{1+}_A \ket{1+}_B,  & ~~ \ket{\xi_{24}} &= \ket{1i_+}_A \ket{1i_-}_B,  & ~~
\ket{\xi_{31}} &= \ket{+0}_A \ket{+0}_B,  & ~~ \ket{\xi_{32}} &= \ket{+1}_A \ket{+1}_B,  \\
\ket{\xi_{34}} &= \ket{+i_+}_A \ket{+i_-}_B,  & ~~ \ket{\xi_{41}} &= \ket{i_+0}_A  \ket{i_-0}_B,  & ~~
\ket{\xi_{42}} &= \ket{i_+1}_A \ket{i_-1}_B, & ~~ \ket{\xi_{43}} &= \ket{i_++}_A  \ket{i_-+}_B
\end{aligned}
\right\}.
\end{equation}
 The above states are already linearly independent, and hence they are conclusively distinguishable using global measurements. We now investigate the status of conclusive local discrimination for this set of states.

We attempt to find a product state that can detect \( \ket{\xi_{12}}_{AB} \).  Assume there exists a product state $\ket{\chi} \in \mathbb{C}^4 \otimes \mathbb{C}^4$, such that \( \ket{\chi} = \ket{\alpha}_{A} \otimes \ket{\beta}_{B} \), where:
\begin{align*}
\ket{\alpha} &= (a_1\ket{00}+a_2\ket{01}+a_3\ket{10}+a_4\ket{11})_{A_1A_2}, \quad \text{with } \sum |a_i|^2 = 1, \\
\ket{\beta}  &= (b_1\ket{00}+b_2\ket{01}+b_3\ket{10}+b_4\ket{11})_{B_1B_2}, \quad ~\text{with } \sum |b_i|^2 = 1.
\end{align*}
, where $\{\ket{i}_A\}$ and $\{\ket{j}_B\}$ respectively denote the computational basis set for subsystems of Alice's and Bob's, respectively. We require that \( \ket{\chi} \) is orthogonal to all \( \ket{\xi_{ij}} \) in Eqn.~\ref{duan} except \( \ket{\xi_{12}} \)
and $\braket{\xi_{12} | \chi} \neq 0 $.\\

This leads to the following system of equations:

\begin{subequations}
\begin{minipage}{.4\linewidth}
\begin{align}
(a_1+a_2)(b_1+b_2) &= 0, \label{eq:1}\\
(a_1-\iota~ a_2)(b_1+\iota~b_2) &= 0, \label{eq:2} \\
a_3 b_3 &= 0, \label{eq:3} \\
(a_3+a_4)(b_3+b_4) &= 0, \label{eq:4} \\
(a_3-\iota~a_4)(b_3+\iota~b_4) &= 0, \label{eq:5} \\
(a_1+a_3)(b_1+b_3) &= 0, \label{eq:6}   
\end{align}
\end{minipage}%
\begin{minipage}{.5\linewidth}
\begin{align}
(a_2+a_4)(b_2+b_4) &= 0, \label{eq:7} \\
(a_1-\iota~a_2+a_3-\iota~a_4)(b_1+\iota~b_2+b_3-\iota~b_4) &= 0, \label{eq:8} \\
(a_1-\iota~a_3)(b_1+\iota~b_3) &= 0, \label{eq:9} \\
(a_2-\iota~a_4)(b_2+\iota~b_4) &= 0, \label{eq:10} \\
(a_1+a_2-\iota~a_3-\iota~a_4)(b_1+b_2+\iota~b_3+\iota~b_4) &= 0, \label{eq:11} \\
a_2b_2 &\neq 0. \label{eq:12}  
\end{align}
\end{minipage}
\end{subequations}

where $\iota:=\sqrt{-1}$. \\

From Eqn.~\eqref{eq:12}, we immediately find that
\begin{equation}
a_2 \neq 0 \quad \text{and} \quad b_2 \neq 0.
\end{equation}

From Eqn.~\eqref{eq:3}, either \( a_3 = 0 \) or \( b_3 = 0 \). Without loss of generality, assume
\begin{equation}
a_3 = 0, \quad b_3 \neq 0.
\end{equation}

From Eqn.~\eqref{eq:1}, assume
\begin{equation}
a_1 = -a_2 \neq 0.
\end{equation}

Substituting into Eqn.~\eqref{eq:2}, we find
\begin{equation}
b_1 = -\iota~b_2 \neq 0.
\end{equation}

From Eqn.~\eqref{eq:9}, either \( a_3 = -a_4 \) or \( b_3 = -b_4 \). Since \( a_3 = 0 \), this gives
\begin{equation}
a_4 = 0.
\end{equation}

Next, from Eqn.~\eqref{eq:11}, using \( a_4 = 0 \), \( a_1 = -a_2 \), and \( b_1 = -\iota~b_2 \), we have
\begin{equation}
b_1 = -b_3 = -\iota~b_2 \neq 0.
\end{equation}

From Eqn.~\eqref{eq:10}, either \( a_2 = -a_4 \) or \( b_2 = -b_4 \). Assume
\begin{equation}
a_4 = -a_2 \neq 0.
\end{equation}
However, from earlier, \( a_4 = 0 \), leading to a contradiction unless \( a_2 = 0 \), which contradicts Eqn.~\eqref{eq:12}.

Finally, using Eqns.~\eqref{eq:8} and \eqref{eq:11}, we find
\begin{equation}
b_3 = 0,
\end{equation}
which again contradicts \( b_2 \neq 0 \) from Eqn.~\eqref{eq:12}.

Thus, we arrive at a contradiction. This contradiction emerges from one particular choice of assumptions. A similar inconsistency arises for all other choices of setting parameters to zero. Therefore, we conclude that the state \( \ket{\xi_{12}} \) cannot even be conclusively identified via LOCC.

Consequently, by Definition~\textcolor{red}{1}, the entire set of $12$ product states cannot be locally and conclusively distinguished.
\end{proof}
\section{Converse of theorem \ref{theo1}}\label{appendixF}
In Theorem \ref{theo1}, we showed that: 
\textit{
    For a set of multipartite states $\mathcal{S}$, conclusive local state discrimination (i.e. $1$-CLSM) always implies m-conclusive local state marking (i.e. $m$-CLSM ~~~$\forall m \in \{1,2,\cdots ,|\mathcal{S}|\}$).}\\

The claim that the product state sets \(\mathcal{S}^{\uparrow\hspace{-.05cm}\downarrow}\) and \(\mathcal{S}^D\), which exhibit \(2\)-CLSM indistinguishability, demonstrate a stronger form of nonlocality than the original constructions by Bennett {\it et al.}~\cite{BennettNLWE} and Peres and Wootters~\cite{PeresWooters}, is supported by the fact that, although Theorem~ \textcolor{red}{1} establishes that CLSD implies \(m\)-CLSM, the converse does not hold—a result we substantiate through explicit counterexamples. Specifically, we show that while the product state sets \(\mathcal{S}^{\uparrow\hspace{-.05cm}\downarrow}\) and \(\mathcal{S}^D\) remain locally conclusively indistinguishable under their $2$-marking, the same does not hold for certain binary ensembles of mutually orthogonal quantum states. These ensembles, despite being locally conclusively indistinguishable, become conclusively markable under LOCC. Notably, such ensembles have previously been regarded as exhibiting a high degree of nonlocality, as they remain locally conclusively indistinguishable even with arbitrarily large but finite copies assistance.

Consider the following binary ensemble of mutually orthogonal quantum states as introduced by Yu {\it et al.} \cite{Duan}:
\[
\mathcal{X}^Y := \left\{\rho_{AB} := \Phi^+_{d}~,~ \sigma_{AB} := \frac{1}{(d^2-1)}(\mathbb{I}_{d \times d} - \Phi^+_{d})\right\}.
\] where $\Phi^+_{d}$ is the projector corresponding to the state $\ket{\Phi^+_{d}}_{AB}:= \frac{1}{\sqrt{d}} \sum_{i=0}^{d-1} \ket{ii}_{AB} \in \mathbb{C}^d_{A} \otimes \mathbb{C}^d_{B} $.
These states are locally indistinguishable. This result was established by proving that even under PPT operations, they remain conclusively indistinguishable. Remarkably, even with access to arbitrarily large but finite copies, the states cannot even be conclusively distinguished. Thus, \[
\mathcal{X}^Y[n] := \left\{\rho^{\otimes n}_{AB},~ \sigma^{\otimes n}_{AB}\right\}.
\] where $n\in \mathbb{N}$, the set of natural numbers is again locally conclusively indistinguishable. 
However, we demonstrate that these states are conclusively markable via LOCC. Hence, 
$
\mathcal{X}^Y_{(2)} := \left\{\rho_{A_1B_1} \otimes \sigma_{A_2B_2}~,~ \sigma_{A_1B_1} \otimes \rho_{A_2B_2}\right\}
$ is locally conclusively distinguishable, which leads us to the following theorem.  
\begin{theorem}\label{theo4}
CLSM of any given set of states $\mathcal{S}$ does not necessarily imply CLSD of $\mathcal{S}$.
\end{theorem}

\begin{proof}

We now describe a protocol that allows Alice and Bob to conclusively locally mark the set $\mathcal{X}^Y_{(2)}$ and thus prove the theorem. They begin by selecting \(\ket{01}\) as the product detecting state corresponding to \(\sigma_{AB}\). A computational basis measurement is performed on their respective first subsystems, and the outcomes are shared through classical communication. If the projector onto the detecting state clicks for both parties, they can conclusively infer that the state is \(\sigma_{A_1B_1} \otimes \rho_{A_2B_2}\). In the case where the measurement yields correlated outcomes, the result is deemed inconclusive. To resolve this ambiguity, Alice and Bob repeat the same procedure on the second subsystem. A conclusive outcome in this round then indicates that the state must have been \(\rho_{A_1B_1} \otimes \sigma_{A_2B_2}\). Therefore, the set  $\mathcal{X}^Y$ is locally conclusively markable, even though it is known that its arbitrarily large but finite copies do not allow CLSD. This completes the proof.
\end{proof} 

This result reveals a striking phenomenon: while individual copies of the states \(\rho_{AB}\) or \(\sigma_{AB}\), even with arbitrarily many copies, remain conclusively indistinguishable under LOCC, their tensor product combinations enable conclusive discrimination. This underscores the power of joint state configurations in overcoming the limitations imposed by LOCC constraints. Consequently, in the context of a marking task, the set of states appears to lose the high degree of nonlocality typically attributed to it.\\

Interestingly, we now consider another binary ensemble of mutually orthogonal quantum states that exhibits the same feature. In \cite{Bandyopadhyay}, it was proved that: 
\setcounter{lemma}{3} 
\begin{lemma}
Any bipartite orthogonal ensemble which
contains $\sigma$, the normalized projector onto a UPB subspace,
is conclusively locally indistinguishable in the many
copy limit    
\end{lemma}

Let us define \(\sigma^{UPB}_{AB}\) as the normalized projection operator onto any bipartite UPB subspace, and let \(\rho^{Ent}_{AB}\) be any quantum state orthogonal to it.
Consider the set:
 $$\mathcal{X}^B:=\{\sigma^{UPB}_{AB}~,~\rho^{Ent}_{AB} \}$$ 
Bandyopadhyay \cite{Bandyopadhyay} showed that the set is conclusively indistinguishable via LOCC, as there exists no product state capable of detecting \(\rho^{Ent}_{AB}\), as that product state needs to be necessarily orthogonal to the UPB subspace, which from the definition of UPB is not possible. Furthermore, since the tensor product of two UPB sets forms another UPB~\cite{DiVincenzo03}, the set remains conclusively indistinguishable even when provided with an arbitrarily large but finite number of copies. However, just like $\mathcal{X}^Y$, even this pair of states is conclusively markable.
 \setcounter{proposition}{4}
 \begin{proposition}\label{propositionUPB}
 The set $\mathcal{X}^B_{(2)} := \left\{\sigma^{UPB}_{A_1B_1} \otimes \rho^{Ent}_{A_2B_2}~,~ \rho^{Ent}_{A_1B_1} \otimes \sigma^{UPB}_{A_2B_2}\right\}
$ is conclusively distinguishable via LOCC.
 \end{proposition}
 \begin{proof}
     Alice and Bob choose one of the UPB state itself as the corresponding product detecting state for $\sigma^{UPB}_{AB}$. Then they perform the measurement with one effect as the projector of this detecting state and the rest as the remaining product states that completes this basis. If the projector corresponding to the chosen detecting state clicks, they are sure that the first state is the state $\sigma^{UPB}_{A_1B_1}$. Hence the second state must be $\rho^{Ent}_{A_2B_2}$. However, if the projector does not click, they regard the outcome as inconclusive and move on to the second state and perform the same measurement with the same line of strategy and inference. If they get a conclusive outcome, then they are certain that the state was  $\sigma^{UPB}_{A_2B_2}$ and thus the first state was $\rho^{Ent}_{A_1B_1}$.
 \end{proof}
It is true that the multicopy conclusive distinguishability of the sets \(\mathcal{S}^{\uparrow\hspace{-.05cm}\downarrow}\) or \(\mathcal{S}^{\text{D}}\) has not yet been fully investigated. Suppose, for instance, that the set \(\mathcal{S}^{\uparrow\hspace{-.05cm}\downarrow}\) becomes conclusively distinguishable when assisted with two copies. However, it is already known that the same set is not conclusively markable in the one-copy assistance scenario. In contrast, the states in the sets \(\mathcal{X}^Y\) and \(\mathcal{X}^B\) remain conclusively indistinguishable via LOCC, even when arbitrarily many copies are provided. Yet, these sets are conclusively markable locally using only a single copy of each state in the permutation. This contrast underscores the fundamental inequivalence between the notions of CLSD and CLSM, especially when examined through the lens of quantum nonlocality.\\

While the previous examples discussed in Theorem~\ref{theo4} and Proposition~\ref{propositionUPB} may suggest that orthogonal ensembles can enable a locally conclusively indistinguishable set to become conclusively markable, this intuition does not hold true. We now present a counterexample involving a set of mutually orthogonal states for which \(2\)-CLSM is not possible. This set \(\mathcal{S}^{\text{Bell}}\), comprises the four Bell states. As discussed before it is well known that these states are locally indistinguishable~\cite{GKar}, and they also remain conclusively indistinguishable under LOCC. Furthermore, it was shown in~\cite{LSM} that these states are locally unmarkable. As a result, none of the \(4\)-LSM, \(3\)-LSM, or \(2\)-LSM tasks are achievable for this set. In this work, we go a step further and demonstrate that \(2\)-CLSM is  not even possible for this ensemble. Hence, \(\mathcal{S}^{\text{Bell}}\) serves as an example of a mutually orthogonal set that is not only locally indistinguishable and unmarkable, but also conclusively unmarkable under the \(2\)-CLSM framework.

\begin{proposition}\label{propositionBell}
Mutual orthogonality does not guarantee conclusive markability under LOCC.
\end{proposition}

\begin{proof}
The sketch of the proof is as follows. We establish the result by demonstrating that the set \(\mathcal{S}^{\text{Bell}}\) does not admit a \(2\)-CLSM. We shall prove this via contradiction.\\ We define the \(2\)-CLSM set of $\mathcal{S}^{\mbox{Bell}}$,  as  $\mathcal{S}^{\mbox{Bell}}_{(2)}$  where 
\begin{equation}
\label{CLSMBell}
\mathcal{S}^{\mbox{Bell}}_{(2)}:=\left\{
\begin{aligned}
&\ket{\gamma}^{12}:=\ket{\mathcal{B}^1}\ket{\mathcal{B}^2},~~~~~\ket{\gamma}^{13}:=\ket{\mathcal{B}^1}\ket{\mathcal{B}^3},\\
&\ket{\gamma}^{14}:=\ket{\mathcal{B}^1}\ket{\mathcal{B}^4},~~~~~\ket{\gamma}^{21}:=\ket{\mathcal{B}^2}\ket{\mathcal{B}^1},\\
&\ket{\gamma}^{23}:=\ket{\mathcal{B}^2}\ket{\mathcal{B}^3},~~~~~\ket{\gamma}^{24}:=\ket{\mathcal{B}^2}\ket{\mathcal{B}^4},\\
&\ket{\gamma}^{31}:=\ket{\mathcal{B}^3}\ket{\mathcal{B}^1},~~~~~\ket{\gamma}^{32}:=\ket{\mathcal{B}^3}\ket{\mathcal{B}^2},\\
&\ket{\gamma}^{34}:=\ket{\mathcal{B}^3}\ket{\mathcal{B}^4},~~~~~\ket{\gamma}^{41}:=\ket{\mathcal{B}^4}\ket{\mathcal{B}^1},\\
&\ket{\gamma}^{42}:=\ket{\mathcal{B}^4}\ket{\mathcal{B}^2},~~~~~\ket{\gamma}^{43}:=\ket{\mathcal{B}^4}\ket{\mathcal{B}^3}  
\end{aligned}
\right\}
\end{equation}
 is written in the $A_1B_1A_2B_2$  party notation.

The above states are mutually orthogonal, and hence globally distinguishable, and are thus trivially conclusively distinguishable using global measurements. We now investigate the status of conclusive local state discrimination for this set of states. According to Definition~\textcolor{red}{1}, every element in the set must be conclusively identifiable, for the set to be conclusively distinguishable via LOCC. 
Thus, for each of the twelve states in the set, a corresponding product detecting state must exist. However, we show that such a product detecting state cannot be found for at least one member of the set, rendering it conclusively unidentifiable. This leads to a contradiction with Lemma~\textcolor{red}{2}, thereby establishing the claim. \\Let us first check whether \( \ket{\xi_{12}}_{AB} \) is conclusively identifiable locally. Using Lemma \textcolor{red}{2}, we attempt to find a product state that can detect it. Assume there exists a product state $\ket{\chi} \in \mathbb{C}^4 \otimes \mathbb{C}^4$, such that \( \ket{\chi} = \ket{\alpha}_{A} \otimes \ket{\beta}_{B} \), with:

\begin{align*}
\ket{\alpha} &= (a_1\ket{00}+a_2\ket{01}+a_3\ket{10}+a_4\ket{11})_{A_1A_2} \quad  \\
\ket{\beta}  &= (b_1\ket{00}+b_2\ket{01}+b_3\ket{10}+b_4\ket{11})_{B_1B_2} \quad
\end{align*}
where $\{\ket{i}_A\}$ and $\{\ket{j}_B\}$ respectively denote the computational basis set for subsystems of Alice's and Bob's, respectively, along with  $\sum |a_i|^2 = 1$, and $\sum |b_i|^2 = 1 $.

Hence, we require that \( \ket{\chi} \) is orthogonal to all \( \ket{\gamma}^{ij} \) in Eq.(\ref{CLSMBell}) except \( \ket{\gamma}^{12} \)  to which
\(
\braket{\xi_{12} | \chi} \neq 0.
\)  This leads to the following system of equations over complex variables
\( a_1, a_2, a_3, a_4, b_1, b_2, b_3, b_4 \in \mathbb{C} \) , which we show has no solution:

\begin{subequations}
\begin{minipage}{.45\linewidth}
\begin{align}
a_1 b_1 + a_2 b_2 - a_3 b_3 - a_4 b_4 &= 0,\label{eq:C1}\\
a_1 b_2 + a_2 b_1 + a_3 b_4 + a_4 b_3 &= 0,\label{eq:C2}\\
a_1 b_3 + a_3 b_1 + a_2 b_4 + a_4 b_2 &= 0,\label{eq:C3}\\
a_1 b_2 - a_2 b_1 + a_3 b_4 - a_4 b_3 &= 0,\label{eq:C4}\\
a_1 b_3 - a_3 b_1 + a_2 b_4 - a_4 b_2 &= 0,\label{eq:C5}\\
a_1 b_2 + a_2 b_1 - a_3 b_4 - a_4 b_3 &= 0\label{eq:C6},   
\end{align}
\end{minipage}%
\begin{minipage}{.45\linewidth}
\begin{align}
a_1 b_3 + a_3 b_1 - a_2 b_4 - a_4 b_2 &= 0,\label{eq:C7}\\
a_1 b_2 - a_2 b_1 - a_3 b_4 + a_4 b_3 &= 0,\label{eq:C8}\\
a_1 b_3 - a_3 b_1 - a_2 b_4 + a_4 b_2 &= 0,\label{eq:C9}\\
a_1 b_4 - a_2 b_3 + a_3 b_2 - a_4 b_1 &= 0,\label{eq:C10}\\
a_1 b_4 - a_3 b_2 + a_2 b_3 - a_4 b_1 &= 0,\label{eq:C11}\\
a_1 b_1 - a_2 b_2 + a_3 b_3 - a_4 b_4 &\neq 0 \label{eq:C12}.   
\end{align}
\end{minipage}
\end{subequations}

We analyze the system of equations and derive contradictions as follows:

From equations (\ref{eq:C2}) and (\ref{eq:C4}):
\begin{align*}
(a_1b_2 + a_3b_4) + (a_2b_1 + a_4b_3) &= 0 \\
(a_1b_2 + a_3b_4) - (a_2b_1 + a_4b_3) &= 0
\end{align*}
Adding/subtracting gives:
\begin{align}
a_1b_2 + a_3b_4 &= 0 \label{eq:pair1} \\
a_2b_1 + a_4b_3 &= 0 \label{eq:pair2}
\end{align}

Similarly for equations (\ref{eq:C6}) and (\ref{eq:C8}):
\begin{align*}
(a_1b_2 - a_3b_4) + (a_2b_1 - a_4b_3) &= 0 \\
(a_1b_2 - a_3b_4) - (a_2b_1 - a_4b_3) &= 0
\end{align*}
Yielding:
\begin{align}
a_1b_2 = a_3b_4 = 0 \label{eq:zeros1} \\
a_2b_1 = a_4b_3 = 0 \label{eq:zeros2}
\end{align}

From equations (\ref{eq:C3}), (\ref{eq:C5}), (\ref{eq:C7}), and (\ref{eq:C9}) through similar analysis:
\begin{align}
a_1b_3 = a_2b_4 = 0 \label{eq:zeros3} \\
a_3b_1 = a_4b_2 = 0 \label{eq:zeros4}
\end{align}

\textbf{Case 1:} $a_1 \neq 0$ \\
From (\ref{eq:zeros1}) and (\ref{eq:zeros3}): $b_2 = b_3 = 0$

\begin{itemize}
\item \textit{Subcase 1a:} $b_1 \neq 0$ \\
From (\ref{eq:zeros2}) and (\ref{eq:zeros4}): $a_2 = a_3 = 0$ \\
Equation (\ref{eq:C1}) becomes $a_1b_1 = a_4b_4$ \\
The non-degeneracy condition (\ref{eq:C12}) becomes:
\begin{equation*}
a_1b_1 - a_4b_4 \neq 0 
 ~(\text{thus giving rise 
 to a contradiction})
\end{equation*}

\item \textit{Subcase 1b:} $b_1 = 0$ \\
Equation (\ref{eq:C1}) reduces to $a_4b_4=0$ \\
Condition (\ref{eq:C12}) becomes $0 - 0 + 0 - a_4b_4 \neq 0  ~(\text{which gives rise 
 to a contradiction})$
\end{itemize}

\textbf{Case 2:} $a_1 = 0$ \\
From equations (\ref{eq:zeros1}-\ref{eq:zeros4}), similar analysis leads to contradictions
in all subcases.

Thus, we see that 
all possible cases lead to  contradictions in the non-degeneracy condition (\ref{eq:C12}). Therefore, no solution exists that satisfies all equations simultaneously.

\end{proof}

Interestingly, it was found in \cite{LSM} that for the above set of states, no LSM is possible. Our proof explains this by showing that local conclusive state discrimination, which is a weaker task than local state discrimination, is itself not possible.

\twocolumngrid

%

\end{document}